\documentclass[runningheads]{llncs}
\usepackage[T1]{fontenc}
\usepackage{graphicx}

\usepackage{cite}
\usepackage{epstopdf}
\usepackage{amsfonts,amssymb}
\usepackage[hyphens]{url}
\usepackage[colorlinks=true,linkcolor=blue,anchorcolor=blue,citecolor=blue,urlcolor=blue, breaklinks=true]{hyperref}
\usepackage{xurl}
\usepackage[ruled,vlined]{algorithm2e}
\usepackage{amsmath}
\usepackage{cancel}
\usepackage[misc]{ifsym}
\usepackage{color}

\begin{document}
\thispagestyle{empty}
\title{Less-excludable Mechanism for DAOs in Public Good Auctions}
%
%
\author{Jing Chen \and Wentao Zhou*}
%
\authorrunning{J. Chen et al.}
%
\institute{Department of Computer Science and Technology, Tsinghua University,\\ Beijing 100084, China \\
\email{jchencs@tsinghua.edu.cn,\ zhouwt24@mails.tsinghua.edu.cn}}

\maketitle              
\begin{abstract}
With the rise of smart contracts, decentralized autonomous organizations (DAOs) have emerged in public good auctions, allowing ``small'' bidders to gather together and enlarge their influence in high-valued auctions.
However, models and mechanisms in the existing research literature do not guarantee {\em non-excludability}, which is a main property of public goods. As such, some members of the winning DAO may be explicitly prevented from accessing the public good.
This side effect leads to regrouping of small bidders within the DAO to have a larger say in the final outcome.
In particular,
we provide a polynomial-time algorithm
to compute the best regrouping of bidders that maximizes the total bidding power of a DAO. We also prove that such a regrouping is less-excludable, better aligning the needs of the entire DAO and the nature of public goods.

Next, notice that members of a DAO in public good auctions often have a positive externality among themselves.
Thus we introduce a {\em collective factor} into the members' utility functions. We further extend the mechanism's allocation for each member to allow for {\em partial access} to the public good.
Under the new model, we propose a mechanism that is incentive compatible in generic games and achieves higher social welfare as well as less-excludable allocations.

\end{abstract}

\section{Introduction}
The combination of blockchains and auctions has become increasingly prevalent in recent years, because blockchain features such as decentralization, data immutability and smart contracts naturally provide convenience for auctions. A representative example is the emergence of Decentralized Autonomous Organizations (DAOs) \cite{buterin2013bootstrapping, buterin2014daos} and their application in public good auctions.

DAO is a new organizational framework in which management and operation rules are encoded in smart contracts on the blockchain. It can unite a group of people for a common goal in a decentralized format, such as fund-raising and auctions \cite{hassan2021decentralized, wang2019decentralized}.
By 2022, the scale of DAOs had reached 1.6 million participants, collectively managing over 13,000 DAOs with a total value of \$16 billion \cite{jenkinson2022remote,tan2023open}.
By 2024, the total treasury of DAOs reached \$37.6 billion \cite{deepdaoOrganizations}, a phenomenon that cannot be ignored in the real world.


\setcounter{page}{1}  
\pagenumbering{arabic}

In both on-chain and off-chain auctions, DAOs can be treated as bidders. Their internal structures allow them to gather more power from individuals and participate in the auction in a more efficient way.
DAOs are especially closely associated with auctions of public goods, as public goods have two key characteristics, {\em non-rivalry} and {\em non-excludability} \cite{kaul1999defining},
which align well with the motivation of DAOs and the permissionless and egalitarian nature of blockchains.
It's thus easy to form DAOs for such purposes. For example, in 2021, when Sotheby's auctioned off one of the original copies of the United States Constitution, more than 17,000 people formed the ConstitutionDAO and raised more than \$47 million to join the auction by smart contracts deployed on Ethereum. Although the ConstitutionDAO narrowly failed to win the auction, this new form of participation has been recognized by traditional auction houses.

The participation of DAOs in an auction requires a two-level mechanism as suggested by \cite{bahrani2023bidders}. The upper-level mechanism runs by the auctioneer and treats each DAO as an individual bidder participating in the auction. The lower-level mechanism runs by each DAO to aggregate the DAO's members and act on behalf of them.
This includes collecting the members' bids to form a total bid that is submitted to the upper-level mechanism and, in case of winning the auction in the latter, distributing the access to the item and the corresponding total payment among members.
The mechanism introduced by \cite{bahrani2023bidders} is incentive compatible (IC) for each member of a DAO and is approximately optimal for the social welfare of the auction. However, it has to give up non-excludability in order to achieve IC and other good properties, which means that even when a DAO wins the auction, some small members of it still cannot get access to the item.
This unfortunate side-effect puts the original intention of DAOs at question and leaves space for small members to form another (smaller) DAO among themselves so that they can collectively have a larger say in the larger DAO.
Although game theory often assumes players' independent behavior, a world of DAOs intrinsically leverages the collective power of small members and
one should take such collective behavior into consideration when studying DAOs.


\subsection{Our Contributions}
As non-excludability is a crucial characteristic of public goods, we try to highlight it in our study and balance it with other conflicting goals. With the model and the mechanism of \cite{bahrani2023bidders}
introduced in Section~\ref{s2}, we shall leverage the collective behavior of a DAO's members in two ways.

In Section~\ref{s3}, from an algorithmic point of view, we consider how the members of a DAO in the original mechanism may regroup themselves to build a better substructure within the DAO.
Interestingly, such a collusive behavior doesn't undermine the power of the entire DAO but cooperatively helps the DAO to make a higher bid, besides improving its non-excludability.
For the grouping problem, we construct an efficient algorithm and show that the optimal subgroups within a DAO can be computed in polynomial time.
An open question is: why stopping at two levels?
If some small members are still excluded from accessing the public good given the subgroups, they may continue this procedure to form a smaller DAO within the subgroups within the original DAO, etc, until they are influential enough in the small group and their power gets amplified by their groups in higher levels.
Indeed, the grouping problem can have many levels. But since people usually do not think that far in their reasoning and actions, a two-layer grouping problem is our main focus in this study and we leave the hierarchical consideration to the future.





In Section~\ref{s4}, from a game-theoretic point of view,
we generalize the standard quasi-linear utility model and add a collective factor to a DAO's members' utility function.
This reflects the fact that members of a DAO in a public good auction may have a positive externality among themselves and care about whether others also get access to the item.
Moreover, we extend the auction's allocation space from binary access
to allowing for partial access, so that some small members who were
completely prevented from accessing the item can now be given access to some extent.
For example, a museum holding the auctioned public good (in case of an antique) may be entirely free to some members of the DAO and free for several months every year to other members.
In reality, the format of the partial access can be flexible depending on the nature of the public good, such as lotteries, time-sharing or space-sharing.
In the generalized model, we show a new mechanism that satisfies individual rationality, budget balance, and equal treatment, which are properties of the previous mechanism in \cite{bahrani2023bidders}, as well as incentive compatibility in {\em generic games}.%
\footnote{Recall that in game theory a generic game is where there is no critical tie.}
Compared with the previous mechanism, our mechanism improves both the non-excludability and the social welfare of the auction.
This is a first step towards achieving full non-excludability when DAOs participate in public good auctions, an open problem that deserves further studies.


\subsection{Related Work}
\noindent\textbf{Public Good Auctions.}
Many previous studies have shown by experiments that when public goods are involved, individuals are influenced by social factors and cannot be entirely modeled as self-interested agents \cite{ledyard1994public, fehr2003nature}.
The literature on the social efficiency of public good auctions is also extensive.
For example, \cite{goeree2005not} proved that a winner-pays-only mechanism was not sufficiently effective, while a full payment mechanism offered certain advantages in the context of public goods fund-raising auctions.
For a non-divisible public good which can be provided by different players at different costs, \cite{Kleindorfer1994} defined a non-cooperative game and proposed an auction mechanism that is efficient, fair and incentive compatible, but does not have a uniquely determined payment for each player. The cooperative form was further studied by \cite{Dehez2013}.
These studies took non-excludability as a given property and focused on other aspects of public good auctions.
We instead highlight the impact of non-excludability and consider to what degree it can be achieved, both algorithmically and via incentive-compatible mechanisms.
Recently \cite{li2023altruism} modified the utility function in binary network public goods games to reflect various prosocial motivations such as altruism and collectivism that influence individual decision making. The collective utilities we consider for DAOs' members are inspired by their work, while our objectives are different.



\smallskip
\noindent\textbf{DAOs in Auctions.}
Public good auctions with DAOs as participants were considered by \cite{bahrani2023bidders} as a two-level mechanism. DAOs have their operation rules encoded in smart contracts and act accordingly in collecting member bids, computing the total bid and computing the final allocation and prices for the members. Thus only the members of DAOs are strategic, not the DAOs in between members and the auctioneer. In order to achieve incentive compatibility, their mechanism has to explicitly forbid some members of the winning DAO from accessing the good, thus violating non-excludability.
In \cite{rachmilevitch2022auctions}, the author considered a similar two-level model but focused on analyzing the equilibrium structure for DAOs participating in first- and second-price auctions, and only studied the scenario where a single DAO competed with an individual bidder.
Moreover, \cite{tan2023open} summarized many open problems about DAOs, including the design and analysis of incentive mechanisms for providing different types of public goods.
Our work builds upon \cite{bahrani2023bidders} and aims instead for less-excludable public good auctions.

\smallskip
\noindent\textbf{Cooperative Game Theory.}
When resolving the confliction between incentive compatibility and non-excludability,
we took inspiration from \cite{aumann1985game} for bankruptcy in cooperative game theory,
even though the problems considered are different.
Finally, the subgroups of DAOs' members resemble in spirit coalition forming in cooperative game theory \cite{Dehez2013, ray2015coalition, chalkiadakis2011computational}, but we took a pure algorithmic approach to compute the optimal subgrouping. It would be interesting to further explore the strategic behavior when forming subgroups and establish deeper connections with concepts such as the core.






\section{Preliminaries} \label{s2}
Using the model in \cite{bahrani2023bidders}, an auctioneer is selling a single public good to DAOs and individual bidders. Since an individual bidder can be seen as a DAO with only one member, in the rest of the paper the auction is conducted with $m$ DAOs as bidders.
For a DAO $G$, we may slightly abuse notations and let $G$ also denote the set of its members.
Let $n_G = |G|$. When $G$ is clear from the context, $n$ is used for short and $G=[n]$.
Each member $i\in [n]$ has a true value $v_i$ for the good, which can also be seen as its individual {\em willingness to pay} (WTP). It makes a bid $b_i$ to the DAO.
Let $v_{-i}$ and $b_{-i}$ be the true value vector and the bid vector for {\em all members of all DAOs except $i$.}


The two-level mechanism $M=(M_u, M_\ell)$ of \cite{bahrani2023bidders} can be described as follows.
The upper-level mechanism $M_u$ is a second-price mechanism run by the auctioneer, where each DAO acts as a bidder, the highest bidder wins the good and pays the second highest bid. Each DAO $G$ makes a bid $WTP_{total, G}$ ---the {\em willingness to pay} by the entire DAO, which is computed by the DAO from its members' individual bids and is not necessarily the sum of the members' true values or bids.
The auctioneer collects all DAOs' bids and then informs each DAO $G$ of its allocation $X_{total, G}\in \{0, 1\}$ and price $P_{total, G}$. Again when $G$ is clear from the context, the subscript $G$ is removed from the notations.

The lower-level mechanism $M_\ell$ is run by each DAO $G$ and contains three parts: an aggregation function $WTP_{total}(b_1,\dots,b_n)$ that summarizes all members' bids into a total bid of the DAO and submits to the upper-level mechanism; an access function $x(b_1,\dots,b_n,X_{total},P_{total})=(x_1,\dots,x_n)$, where $\forall i \in [n]$, $x_i\in\{0,1\}$ represents whether member $i$ receives access to the good, with $x_i=0$ whenever $X_{total}=0$; and
a cost-sharing function $p(b_1,\dots,b_n,X_{total},P_{total})=(p_1,\dots,p_n)$, where $\forall i \in [n]$, $p_i$ represents the amount that member $i$ should pay.
Notice that the public good can only be allocated to at most one {\em winning DAO} by the upper-level mechanism, but all members of the winning DAO are able to access it unless explicitly excluded from doing so.
The utility function for each member $i$ can be written as $u_i=v_i\cdot x_i-p_i$. The social welfare of the auction can be written as $SW=\sum_{i\in[n]} v_i\cdot x_i$, where the sum is taken over the members of the winning DAO.
The optimal social welfare is defined as $OPT\_SW = \max_G \sum_{i\in G} v_i$, which assumes access by all members of the DAO that receives the good, inline with non-excludability.

The mechanism design problem focuses on the lower-level mechanism, more specifically the three functions $WTP_{total}, x$, and $p$.
Here are some ideal properties for a mechanism to pursue: for every DAO $G$ and for all instances of the auction,

\begin{itemize}
\item[1.] {\em Incentive compatibility:} $u_i(v_i,b_{-i})\geq u_i(b_i,b_{-i})$, $\forall i\in[n]$, $\forall b_i$, $b_{-i}$.

\item[2.] {\em Budget balance:} $\sum_{i\in [n]}{p_i} = P_{total}$.

\item[3.] {\em Individual rationality:} $u_i(v_i,b_{-i})\geq 0$, $\forall i\in[n]$, $\forall b_{-i}$.

\item[4.] {\em Equal treatment:} if $ b_i = b_j$, then $x_i=x_j$ and $p_i=p_j$, $\forall i,j \in [n]$, $i\neq j$.


\item[5.] {\em Non-excludability:} if $X_{total}=1$ then $x_i=1$ $\forall i\in [n]$.
\end{itemize}

The mechanism $M_\ell$ of \cite{bahrani2023bidders} satisfies all properties above except non-excludability, and approximately maximizes the social welfare.
More specifically, given a DAO $G$, without loss of generality the members are ordered according to their bids non-increasingly, so that $b_1\geq b_2 \geq \cdots \geq b_n$.
According to $M_\ell$,
$G$'s bid is defined as $WTP_{total}=\max_{i\in [n]}\{i\cdot b_i\}$. If $G$ wins in the upper mechanism (i.e., $X_{total}=1$), then find $i^*=\max_{i\in [n]}\{i|i\cdot b_i\ge P_{total}\}$.
The access function and the cost-sharing function are as follows: $\forall i\in [n]$,

\begin{equation}
    x_i = \begin{cases}
        1, & i\le i^*\\
        0, & i>i^*
    \end{cases}\quad \quad \mbox{and} \quad \quad
    p_i = \begin{cases}
    \frac{P_{total}}{i^*}, & i\le i^*\\
    0,  & i>i^*
    \end{cases}
    .
\end{equation}
If DAO $G$ does not win the auction, then $x_i=0$ and $p_i=0$ $\forall i \in [n]$.
As shown by \cite{bahrani2023bidders}, this mechanism
is an $H_\ell$-approximation to the optimal social welfare, namely $\frac{OPT\_SW}{SW}\leq H_\ell$ for all instances of the auction and truthful bids, where $\ell =\max_G |G|$ is the size of the largest DAO and $H_\ell=\sum_{i=1}^\ell\frac{1}{i}$ is the $\ell$-th harmonic number.

The above highlights the advantages of the mechanism.
However, its access function indicates that some members with low bids in the winning DAO may still be prevented from accessing the public good, giving up non-excludability.
This makes it attempting for some members of a DAO to gather together to form subgroups, collectively pretending to be a single member of the DAO in order to gain greater influence in the outcome of the auction; see Fig.~\ref{fig1} for such an example.
In particular, the other key characteristic of public goods, non-rivalry, indicates that
the good can be consumed by multiple people without reducing the amount available to others.%
\footnote{For example, an antique in a museum can be seen by one visitor without reducing how it shows up to other visitors.}
Thus excluding some members from access is
only a mean to achieve incentive compatibility and other properties above.
Examining the problem from a different angle and highlighting the nature of public goods, we consider how non-excludability can be improved while other properties may be weakened.

\begin{figure}
    \centering
    \includegraphics[width=0.48\textwidth, keepaspectratio]{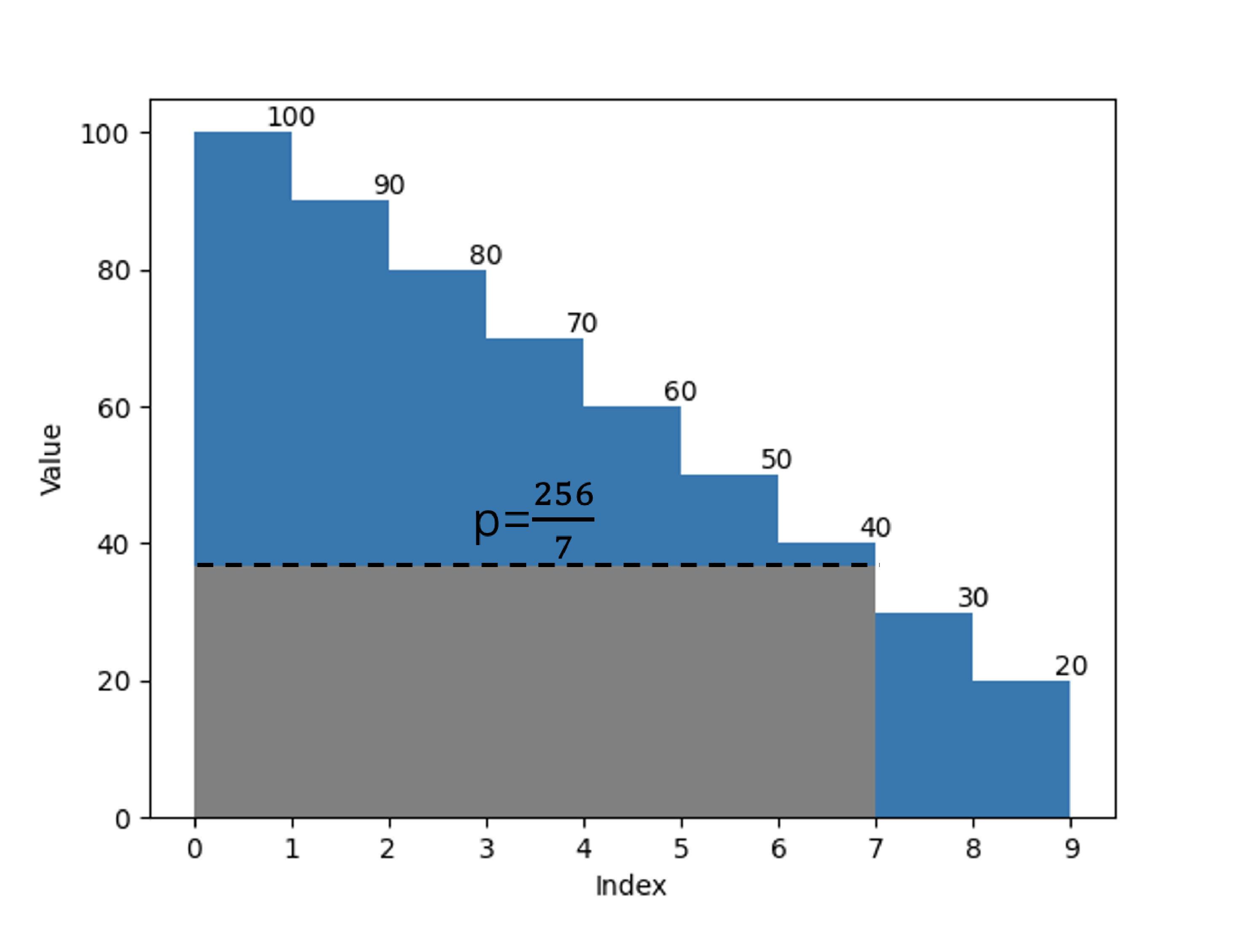}
    \includegraphics[width=0.48\textwidth, keepaspectratio]{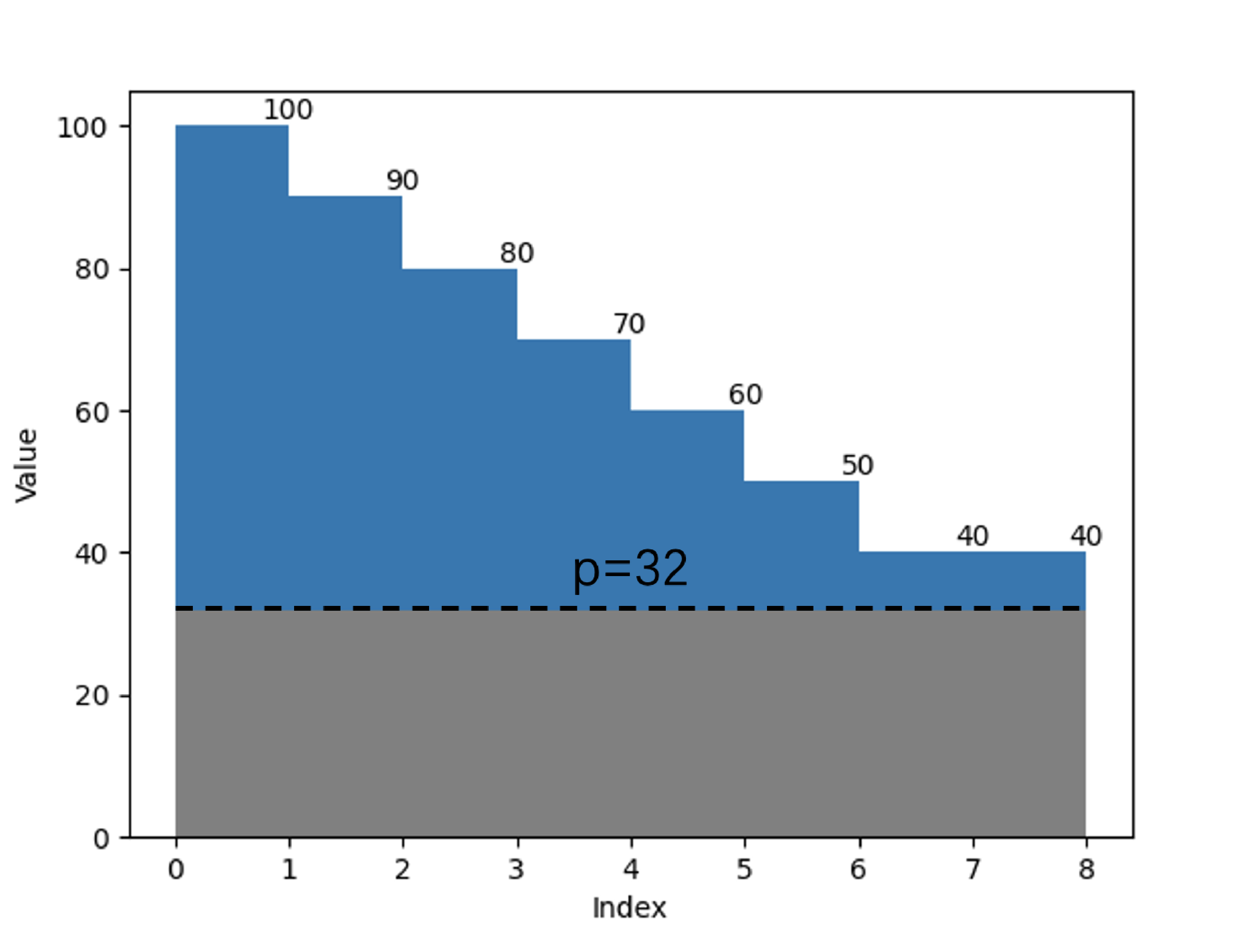}
    \caption{In this example, a DAO has 9 members with valuations 100, 90, $\dots$, 20. Thus its $WTP_{total}=300$.
    If it wins in the upper-level second-price auction with $P_{total}=256$,
    then according to the mechanism, $i^*=7$ and $p_i = \frac{256}{7}$ $\forall i\leq 7$.
    As the left figure shows, the vertical axis represents the amount of bid and the horizontal axis represents the member index.
    The shaded gray area represents $P_{total}=i^*\cdot \frac{P_{total}}{i^*}$.
    The last two members with valuations of 30 and 20 have no access to the good despite of being part of the DAO.
    However, if these two members form a subgroup with willingness to pay $WTP'=\max\{1\times30, 2 \times 20\}=40$,
    then as the right figure shows, $i^*=8$ and $WTP'_{total} = 320$, actually improving the entire DAO's willingness to pay.
    With the same price $P_{total}$, this subgroup can get access and the payment is 32.
    Back in the subgroup, members can now get access and the payment is 16 each, with non-excludability improved.} \label{fig1}
\end{figure}

\section{Subgroups in a DAO and Their Organizations}\label{s3}
From an algorithmic perspective, in this section we consider how the members of a DAO may divide themselves into subgroups (e.g., some of them may form sub-DAOs)
to gain more access to the public good. As we'll show,
such a collective behavior improves both non-excludability and the bidding power of the entire DAO.
%
%
%
Since we are only interested in algorithms instead of strategic mechanisms here,
we consider all members to bid at their true values, that is, $b_i=v_i$. Wlog, the members of a DAO are ordered according to their values non-increasingly, namely $v_1\geq v_2\geq \cdots\geq v_n$.
Below we first introduce the model for a two-level DAO structure with subgroups and use an example to show that the DAO's total willingness to pay can be greatly improved by subgroups.


\subsection{Two-Level DAOs}\label{s3.1}
Inside a DAO $G$, its $n$ members can be further partitioned into $k$ subgroups that are disjoint from each other, referred to as a {\em grouping} of $G$ and denoted by $\mathcal{G}=\{G^1,\dots,G^k\}$, where $\bigcup_{j\in[k]}{G^j}=G$ and $G^j \cap G^{j'}=\emptyset$ $\forall j\neq j'$.
Let $n^j = |G^j|$ for each $j\in [k]$.
When focusing on a subgroup $G^j$, its members' valuations are also referred to as $v_1^j, \dots, v_{n^j}^j$.
Again wlog, the members of $G^j$ are ranked according to their valuations non-increasingly,
namely $v_1^j\geq v_2^j\geq \cdots \geq v_{n^j}^j$.
So a member $a\in [n]$ has the $a$-th highest valuation among all members of $G$ and, if $a$ belongs to the subgroup $G^j$ and has the $t$-th highest valuation in $G^j$, then $v_a = v_t^j$.

Each subgroup's valuation (or rather, its willingness to pay) is determined by the valuations of its members, and the DAO's valuation/willingness-to-pay is determined by those of its subgroups. Equal treatment will be maintained across all subgroups and within each subgroup.
More specifically, for each subgroup $G^j$, its valuation is $WTP^j=\max_{i\in [n^j]}\{i\cdot v_i^j\}$.

From the DAO's perspective, there are $k$ ``members'', $G^1,\dots,G^k$, whose valuations are $WTP^1,\dots,WTP^k$. Again these subgroups are ordered according to their valuations non-increasingly, namely $WTP^1\geq WTP^2 \geq \cdots\geq WTP^k$. The entire DAO's willingness to pay is now $WTP_{total}=\max_{j\in[k]}\{j\cdot WTP^j\}$, also written as $WTP_{total}(\mathcal{G})$ in order to emphasize the grouping $\mathcal{G}$.

If DAO $G$ wins the item at price $P_{total}$ with bid $WTP_{total}$, then for each of the two levels in the DAO, the access rule and the cost-sharing rule follow the original mechanism.
That is, firstly the subgroups of $\mathcal{G}$ get accesses and cost-shares according to the original mechanism as if they were individual members of the DAO. Then for each subgroup $G^j$ that has access to the item with a price~$p$, the same rules apply to its members as if they form a separate DAO.

As the example in Fig.~\ref{fig2} shows, by forming subgroups, more members can contribute and the willingness to pay of the entire DAO may be greatly improved.
Let the maximum willingness to pay of DAO $G$ over all possible groupings be
$opt\_WTP = \max_{k=1}^n \max_{\mathcal{G}=\{G^1, \dots, G^k\}} WTP_{total}(\mathcal{G})$.

In Sections~\ref{s3.3} and~\ref{s3.4}, we show that the optimal grouping whose willingness to pay reaches $opt\_WTP$
enjoys a nice structure and provide a polynomial time algorithm for it.
We prove the optimality of the algorithm and show that it improves non-excludability compared with
the original DAO without subgroups.

\begin{figure}[htbp]
\includegraphics[width=\textwidth, keepaspectratio]{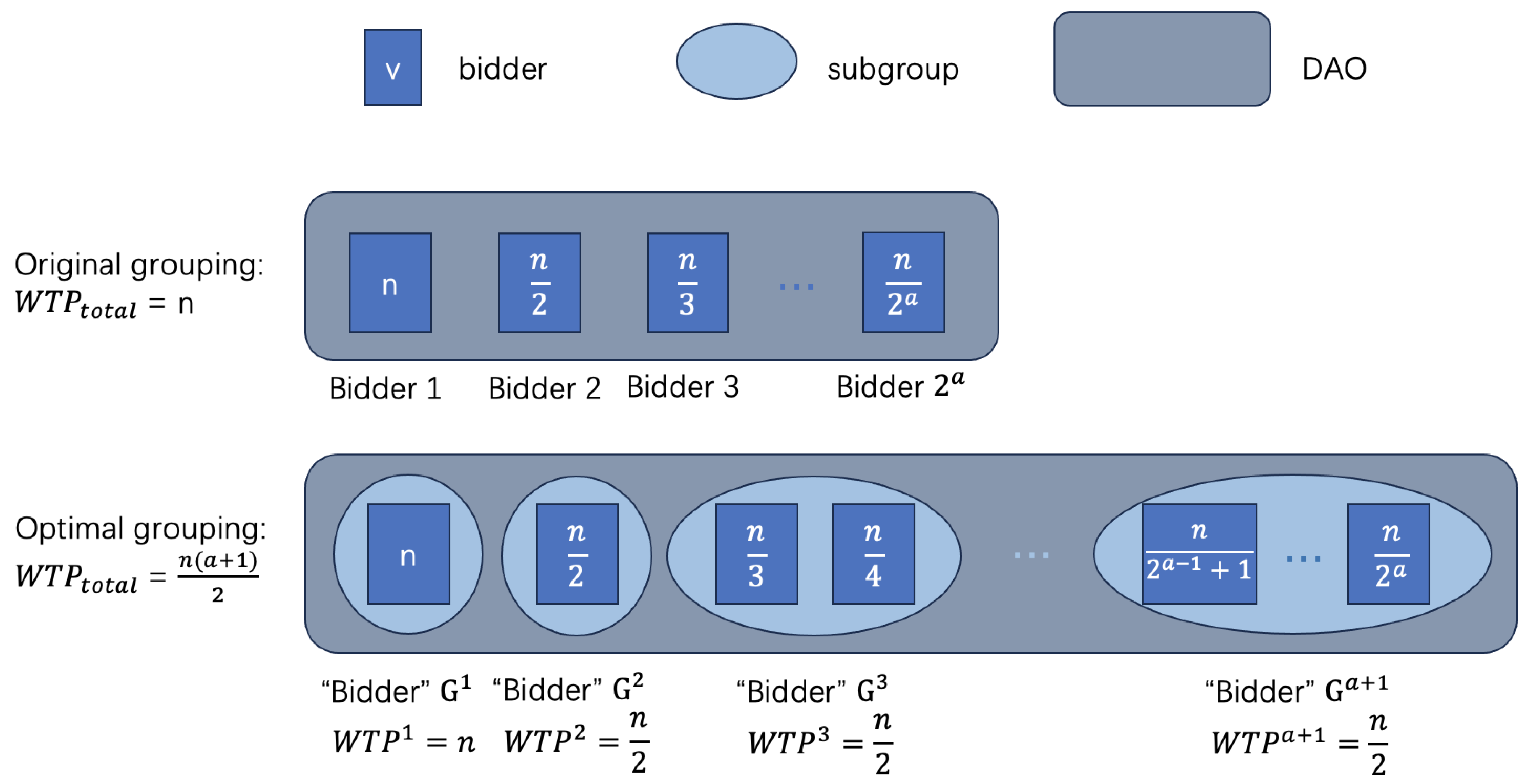}
\caption{
In this example, the DAO has $n=2^a$ members with valuations of $n, \frac{n}{2}, \frac{n}{3}, \dots, \frac{n}{2^a-1}, \frac{n}{2^a}$. Without any subgroup, $WTP_{total}=n$. If the members are partitioned into the following $a+1$ subgroups, $\{1\}, \{2\}, \{3, 4\}, \dots, \{2^{a-1}+1, \dots, 2^a=n\}$, with corresponding valuations $\{n\}, \{\frac{n}{2}\}, \{\frac{n}{3}, \frac{n}{4}\}, \dots, \{\frac{n}{2^{a-1}+1}, \dots, \frac{n}{2^a}\}$, then the total willingness to pay of the DAO becomes $\frac{n(a+1)}{2}$, improved by a logarithmic factor.}
\label{fig2}
\end{figure}

\subsection{The Optimal Grouping Algorithm}\label{s3.3}
We first identify some properties for the optimal grouping.

\begin{definition}
Given a grouping $\mathcal{G}=\{G^1, \dots, G^k\}$, if for some $G^i$, $G^j \in \mathcal{G}$, $G^i\ne G^j$, $\exists$ members $a, c\in G^i$, $ b \in G^j$,
such that the indices satisfy $a<b<c$, then we say that the grouping $\mathcal{G}$ has a \textbf{crossing}.
A grouping $\mathcal{G}$ is \textbf{continuous} if it does not have any crossing.
\end{definition}

Notice that a continuous grouping $\mathcal{G}$ would look like this:
$\{1, 2, \dots, N^1\}, \{N^1 +1, N^1+2, \dots, N^2\}, \dots, \{N^{k-1}+1, N^{k-1}+2, \dots, N^k=n\}$, where $N^j = n^1+n^2+\cdots+n^j$. That is, it consists of $k$ continuous intervals for the members of the DAO.
A contiguous grouping enjoys two natural ways to order its subgroups. On the one hand, the subgroups can be ordered according to their own valuations non-increasingly like before, which makes it easy to compute $WTP_{total}(\mathcal{G})$. On the other hand, they can be ordered by their members' valuations non-increasingly (i.e., by their members' indices non-decreasingly), which makes it easy to describe and analyze. By default the analysis below goes with the first ordering, unless explicitly mentioned.

\begin{theorem}\label{thm1}
For any DAO $G$ and any grouping $\mathcal{G}=\{G^1, \dots, G^k\}$ of $G$, there exists a continuous grouping $\mathcal{\overline{G}}=\{\overline{G}^1, \dots, \overline{G}^k\}$
such that $WTP_{total}(\mathcal{\overline{G}}) \geq WTP_{total}(\mathcal{G})$.
\end{theorem}

Theorem~\ref{thm1} is proved in Appendix~\ref{app}.
We immediately have the following.

\begin{corollary}\label{cor1}
For any DAO $G$, there exists a continuous grouping $\mathcal{G}$ that achieves its maximum willingness to pay, namely $WTP_{total}(\mathcal{G}) = opt\_WTP$.
\end{corollary}

By Corollary~\ref{cor1}, it suffices to consider continuous groupings in order to compute $opt\_WTP$.
Our Algorithm~\ref{alg1} below indeed does so.
To intuitively see how it works, given a grouping $\mathcal{G}$,
we further define the {\em critical index} of $\mathcal{G}$ as $CI = \arg\max_{j\in [k]}\{j\cdot WTP^j\}$, breaking ties in favor of the largest $j$. The corresponding valuation $WTP^{CI}$ is referred to as the {\em critical value} of $\mathcal{G}$ and denoted by $CV$. Thus $WTP_{total}=CI\cdot CV$.

Algorithm~\ref{alg1} enumerates {\em all possible subgroups} and their valuations (the subgroup's willingness to pay).
The key is that, although the number of possible groupings is exponential, the number of distinct subgroups in all continuous groupings is only $O(n^2)$,
and their valuations are all the possible critical values for the optimal continuous grouping.
For each such valuation, assuming for a moment it is the critical value $CV$ of the optimal grouping,
we use a greedy algorithm to form as many subgroups as possible whose valuations are greater than or equal to $CV$. The number of such subgroups is then the critical index $CI$ of the resulting grouping,
and $WTP_{total} = CI\cdot CV$ is a possible value for $opt\_WTP$. We find $opt\_WTP$ by taking the maximum among all $WTP_{total}$'s, and then the corresponding $CV$ leads to
the optimal grouping $\mathcal{G}$.

\begin{algorithm}[h!]
	\small
	\caption{Optimal Grouping Algorithm}
	\label{alg1}
	\LinesNumbered 
	\KwIn{$v_1, \dots, v_n, \text{where}\ v_1 \ge \cdots \ge v_n$}
	\KwOut{$opt\_WTP,\ k,\ \mbox{and the optimal grouping}\  \mathcal{G} = \{G^1, \dots, G^k\}$}
	initialize $possible\_CV\leftarrow\emptyset$\;
	\For{$i=1,\dots,n$}{
		initialize $WTP\leftarrow0$\;
		\For{$j=i,\dots,n$}{
			update $WTP\leftarrow\max\{WTP,\ (j-i+1)\cdot v_j\}$\;
			add $WTP$ to $possible\_CV$;  \tcp{The valuation of $\{i, i+1, \dots,j\}$.}
		}
	}
	initialize $opt\_WTP\leftarrow0,\ opt\_CV\leftarrow0,\ k\leftarrow0$\;
	\For{$CV$ in $possible\_CV$}{
		initialize $CI\leftarrow0$, $subgroup\_size\leftarrow1$\;
		\For{$i=1, \dots, n$}{
			\eIf{$subgroup\_size\cdot v_i \ge CV$}{
				update $CI\leftarrow CI+1$, $subgroup\_size\leftarrow1$\;
			}{
				update $subgroup\_size\leftarrow subgroup\_size+1$\;
			}
		}
		update $WTP_{total}\leftarrow CI\cdot CV$\;
		\If{$WTP_{total} > opt\_WTP$}{
			update $opt\_WTP\leftarrow WTP_{total}$, $opt\_CV\leftarrow CV$, $k\leftarrow CI$\;
		}
	}
	\tcp{Given $opt\_CV$, find the optimal grouping}
	initialize $N\leftarrow[0]$; \tcp{A list with one element 0}
    initialize $subgroup\_size\leftarrow1$\;
	\For{$i=1, \dots, n$}{
		\eIf{$subgroup\_size\cdot v_i \ge opt\_CV$}{
			append $i$ to the end of $N$ and update $subgroup\_size\leftarrow1$\;
		}{
			update $subgroup\_size\leftarrow subgroup\_size+1$\;
		}
	}
	\tcp{Given the list $N$, construct $k$ continuous subgroups, while the remaining members that cannot form an extra subgroup whose valuation reaches $opt\_CV$ is added to the $k$-th subgroup.}
	\For{$j=1, \dots, k$}{
		$\hat{G}^j\leftarrow\{N[j-1]+1, N[j-1]+2, \dots, N[j]\}$\;
	}
	$\hat{G}^{k+1}\leftarrow\{N[k]+1, N[k]+2, \dots, n\}$\;
	$\mathcal{G}\leftarrow\{\hat{G}^1, \dots, \hat{G}^{k-1},\ \hat{G}^{k}\cup \hat{G}^{k+1}\}$\;
\end{algorithm}


\subsection{Analysis of the Optimal Grouping Algorithm}\label{s3.4}
We have the following theorem, proved in Appendix~\ref{app}.

\begin{theorem}\label{thm2}
Given any DAO $G$, Algorithm~\ref{alg1} finds a continuous grouping $\mathcal{G}$ whose valuation is $opt\_WTP$ with time complexity $O(n^3)$.
\end{theorem}

The optimal grouping guarantees that the DAO's $WTP_{total}$ is maximized. Compared to the ungrouped case, it helps the DAO to bid higher and win in more scenarios over other DAOs.
We further show that even if the original ungrouped $WTP_{total}$ is enough to win the item, the optimal grouping guarantees that its allocation is less-excludable compared to the original one.
\begin{definition}[less-excludable allocation]
For a winning DAO $G$ and its two possible allocations $x=(x_1, \dots, x_n)$ and $x'=(x_1', \dots, x_n')$, if $x_i\geq x_i'$ $\forall i\in[n]$,  
then $x$ is {\bf less-excludable} than $x'$.
\end{definition}

Under the current model, $x_i\in\{0,1\}$ $\forall i\in[n]$ and the comparison can be further made on groupings. The ungrouped case is a natural degenerated grouping $\mathcal{G}^*$ with a single subgroup $\{1, 2, \dots, n\}$, and the individually grouped case $\{\{1\}, \{2\}, \dots, \{n\}\}$ is another degenerated grouping which is equivalent to $\mathcal{G}^*$.

\begin{definition}[less-excludable grouping]
For a DAO $G$ and two groupings $\mathcal{G}_1,\ \mathcal{G}_2$, under both of which $G$ is the winning DAO, denote the allocation sets as $A_1,\ A_2$, representing the sets of bidders who have access to the item under the groupings, respectively. If $A_1\supseteq A_2$, then we say that $\mathcal{G}_1$ is \textbf{less-excludable} than $\mathcal{G}_2$.
\end{definition}

The following theorem is also proved in Appendix~\ref{app}.

\begin{theorem}\label{thm3}
For a winning DAO $G$,
the optimal grouping $\mathcal{G}$ by Algorithm~\ref{alg1}
is less-excludable than the degenerated grouping $\mathcal{G}^*$.
\end{theorem}



Notice that there is no worst-case gap between grouping and non-grouping in terms of their WTP, because non-grouping is a special case of grouping and in some scenarios $opt\_WTP$ is achieved in such a case. However, it's still interesting to see, as a future study, for which class of instances a provable lower-bound exists on how much grouping improves over non-grouping.

It is easy to see that the two-level DAO with the optimal grouping satisfies budget balance and individual rationality. It also maintains equal treatment among subgroups and within the same subgroup, as well as allowing more access for small bidders.
However, it doesn't satisfy incentive compatibility as a bidder with a high valuation may benefit by underbidding ---by doing so it may end up in a larger subgroup and pay a lower price.
%
%
%
This reflects the common free-rider problem when considering public goods, where individuals have the incentive to directly enjoy the access to the public good without making sufficient contributions to the collective \cite{albanese1985rational, andreoni1988free}.
In the following section, we consider whether improving non-excludability and achieving incentive compatibility can be better aligned in a different model.

\section{Collective Utilities and a New Mechanism for DAOs}\label{s4}



Recall that the motivation of DAOs is to leverage the collective strength of the members to achieve a goal based on a shared vision.
Especially towards public goods, prosocial behaviors of individuals should be considered, as suggested by \cite{li2023altruism} and \cite{wang2012evolution}.
Therefore we introduce a positive externality into every member's utility who has access to the item, the social welfare  $SW=\sum_{i\in[n]} v_i\cdot x_i$ of the DAO, scaled by a factor $\alpha\geq 0$.
However, a player $i$ may not want to pay more than its true valuation $v_i$ no matter how much it cares about others. Thus a budget constraint is needed to reflect the fact that $v_i$ is really player $i$'s ``willingness to pay''.
In addition, we enrich the access space by allowing partial accesses to the good, so that for each member $i$ of a DAO, 
$x_i\in[0,1]$ instead of $\{0,1\}$.
Accordingly, for each member $i\in[n]$,
\begin{equation*}
	u_i = \begin{cases}
		v_i\cdot x_i-p_i+\alpha\cdot SW, & x_i>0 \mbox{ and } p_i\leq v_i\\
        -\infty, & p_i>v_i \\
		-p_i, & x_i=0
	\end{cases}
	.
\end{equation*}

Note that the social welfare $SW$ also includes $v_i\cdot x_i$, so the utility function guarantees that for any $\alpha$, player $i$ cares more about its own value than it cares about others. Moreover, player $i$ doesn't care about whether others get access or not if itself has no access (that is, $x_i=0$).
Other ways of bringing in a collective factor are of course possible and worth examining in future studies.



\subsection{A New Mechanism} \label{s4.2}
Below we introduce a new lower-level mechanism $M_1$, with an example in Fig.~\ref{fig3} illustrating how it works more intuitively. The upper-level mechanism $M_u$ remains to be the second-price mechanism.
%
Given a DAO $G$ and its members' bids $b_1,\dots, b_n$ ordered non-increasingly, the aggregation function of $M_1$ is
\begin{equation*}
	WTP_{total}(b_1, \dots, b_n)=\max_{i\in [n]}\{\sum_{j=1}^{i}{\min\{b_j,\ (1+\alpha)\cdot b_i}\}\},
\end{equation*}
which is computable in $O(n^2)$ time.

If DAO $G$ does not win the auction in $M_u$, then $x_i=0$, $p_i=0$ $\forall i\in[n]$.
If DAO $G$ wins the auction with price $P_{total}$, we have $WTP_{total}\geq P_{total}$. $M_1$ will compute two payments, $P_1$ and $P_2=\frac{P_1}{1+\alpha}$. Let  $k_{1}=\max_{i\in [n]}\{i|b_i\ge P_1\}$ and $k_{2}=\max_{i\in [n]}\{i|b_i\ge P_2\}$, with $k_1\leq k_2$. They denote the numbers of members who gain full access and who have at least some partial access, respectively.
Given a group of parameters $(P_1, P_2, k_1, k_2)$, the access and the cost-sharing functions are as follows: $\forall i \in [n]$,
\begin{equation*}
	x_i = \begin{cases}
		1, & i\le k_1\\
		\frac{b_i}{P_1}, & k_1<i\leq k_2\\
		0, & i>k_2
	\end{cases}\quad \quad  \mbox{and}\quad \quad
	p_i = \begin{cases}
		P_1, & i\le k_1\\
		b_i, & k_1<i\leq k_2\\
		0,  & i>k_2
	\end{cases}
	.
\end{equation*}

Note that a uniform pricing standard is guaranteed even for partial accesses: $\forall i\in[n]$, $x_i=\frac{p_i}{P_1}$.
The restriction on small members' participation after $k_2$ helps avoid exacerbating free-riding behavior, which could negatively impact the overall effectiveness of the mechanism. The necessity of such a restriction has also been mentioned in works like \cite{goeree2005not}.



A group of parameters $(P_1, P_2, k_1, k_2)$ is {\em feasible} if it satisfies budget balance, $\sum_{i\in [n]} p_i = P_{total}$.
Algorithm~\ref{alg3} below is used by $M_1$ to compute the optimal feasible parameters with the least excludability.
The time complexity of the algorithm is $O(n^2)$.
As shown in Appendix~\ref{app2}, it actually achieves the highest utility for every member among all feasible parameters.

\begin{algorithm}[h!]
	\caption{Optimal Parameters}
	\label{alg3}
	\LinesNumbered 
	\small
	\KwIn{$b_1, \dots, b_n, \text{where}\ b_1 \ge \cdots \ge b_n, \text{and} \ P_{total}$}
	\KwOut{$P_1,\ P_2,\ k_{1},\ k_{2}$}
	initialize $k_{2}\leftarrow n+1$, $size\leftarrow0$\;
	\While{$P_{total}>size$}{
		update $k_2\leftarrow k_2-1$, $size\leftarrow0$\;
		\For{$j=1, \dots, k_2$}{
			update $size \leftarrow size+\min(b_j, (1+\alpha)b_{k_2})$\;
		}
	}
	initialize $P_{remain}\leftarrow P_{total}-k_{2}\cdot b_{k_{2}}$\;
	\If{$P_{remain}\leq 0$}{
		update $k_{1}\leftarrow k_{2}$, $P_1\leftarrow \frac{P_{total}}{k_{2}}$\;
	}
	\Else{
		update $k_{1}\leftarrow k_{2}-1$\;
		\While{$P_{remain} > 0$}{
			\If{$P_{remain}\leq k_{1} \cdot (b_{k_{1}}-b_{k_{1}+1})$}{
				update $P_1\leftarrow b_{k_{1}+1}+\frac{P_{remain}}{k_{1}}$\;
				\textbf{break}\;
			}\Else{
				update $P_{remain} \leftarrow P_{remain}-k_{1} \cdot (b_{k_{1}}-b_{k_{1}+1})$, $k_{1}\leftarrow k_{1}-1$\;
			}
		}
	}
	$P_2 \leftarrow \frac{P_1}{1+\alpha}$\;
\end{algorithm}

\begin{figure}
	\centering
	\includegraphics[width=0.5\linewidth]{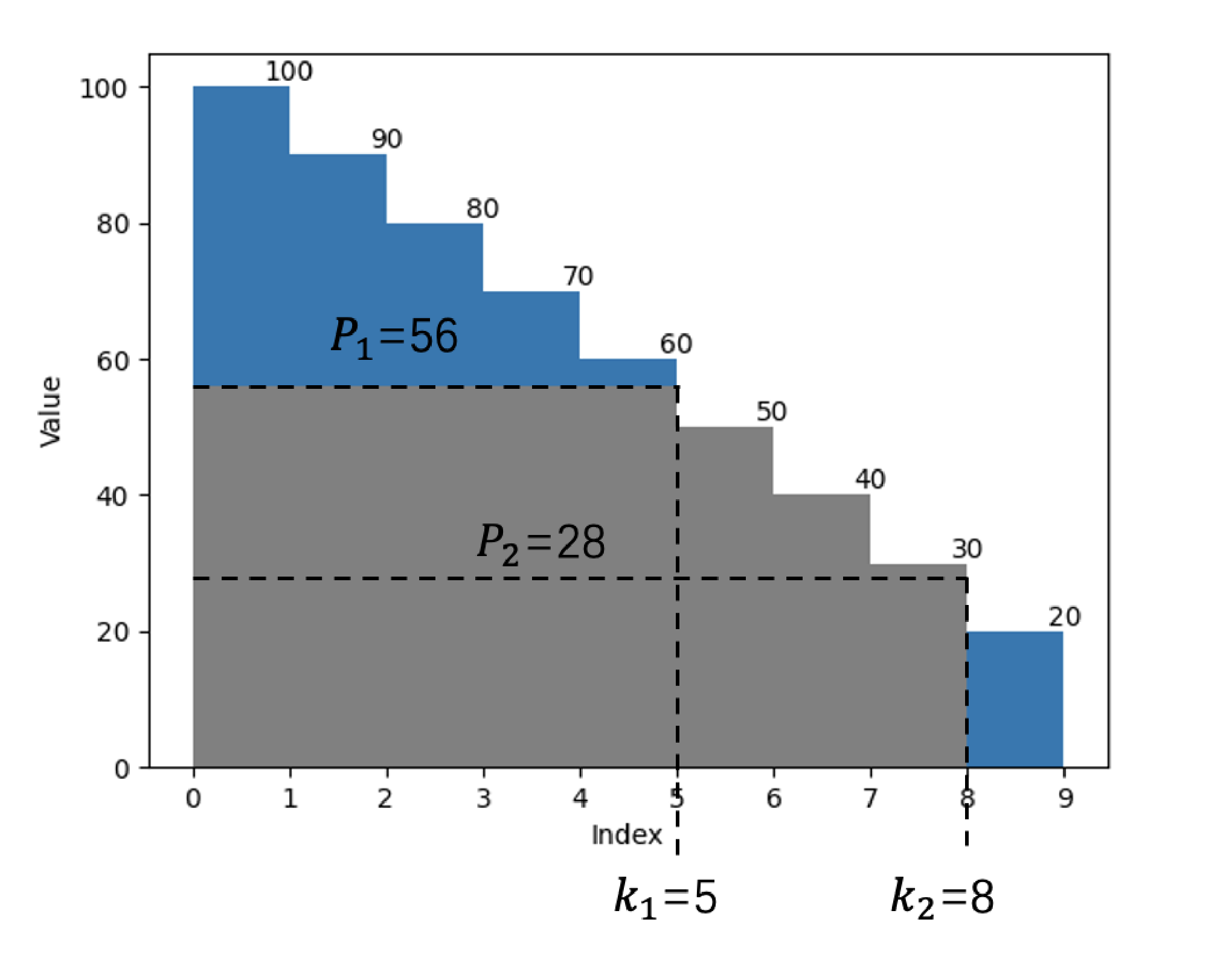}
	\caption{An example for the mechanism $M_1$, with $\alpha=1$ and members in the DAO being the same as Fig.~\ref{fig1}. According to the aggregation function, $WTP_{total}=460$. Assume the DAO wins the auction and $P_{total}=400$. The gray area is used to determine the access and the payment of each member. It can also be viewed as a pool, with the horizontal axis representing the width and the vertical axis representing the height. The capacity of the pool (the size of the gray area) should equal the total payment $P_{total}$ of the DAO and $P_1$ can be seen as the water line.
    The intuition of 
    Algorithm~\ref{alg3} is to find the {\em widest} pool whose capacity reaches $P_{total}$ with the required parameters.}
	\label{fig3}
\end{figure}

\subsection{Analysis of The Mechanism}\label{s4.3}

Mechanism $M_1$ has another advantage:
the value of $\alpha$ in the model can be any nonnegative number without affecting the mechanism's operation.
In fact, $\alpha$ has a practical meaning and can be seen as a measure of the degree of collectivism.
When $\alpha=0$, $M_1$ degenerates to $M_\ell$ and has all its properties such as incentive compatibility. In this case, each member only cares about its individual access and the final allocation only has binary accesses.
As $\alpha$ increases, the proportion of the collective factor in the utility function becomes larger, and the threshold for members to get access becomes lower. This results in
more bids from members being aggregated, leading to a higher overall bid.
For any strictly positive $\alpha$, we have the following main theorem for mechanism $M_1$,
proved in Appendix~\ref{app2}.

\begin{theorem}\label{thm4}
Under the collective-utility model and $\forall \alpha> 0$, mechanism $M_1$ satisfies individual rationality, budget balance, and equal treatment.
It is also incentive compatible in generic games.
Given a winning DAO $G$,
$M_1$ achieves a higher willingness to pay for $G$, a higher social welfare, and a less-excludable allocation compared with $M_\ell$.
\end{theorem}

\section{Future Directions}\label{s6}



Our study tackles non-excludability of public goods for DAOs from an algorithmic perspective and a mechanism design perspective. From both angles we show that the collective behavior of a DAO's members has a rich structure and could improve the social efficiency and the non-excludability of the DAO when treated properly.
Many open problems regarding to non-excludability are also interesting, besides the ones already mentioned in the paper.
For instance, our mechanism can be extended to repeated games by introducing a time factor and relaxing the budget constraint to be time-dependent, transforming the setting into a multi-stage game with observations. Through myopic best-response analysis, we find that the system frequently enters cycles and exhibits stable states.
Also, analyzing the mechanism via solution concepts other than incentive compatibility (e.g., regrets), and considering different auction settings such as multi-item Vickrey auctions for public goods are all intriguing avenues for future research.

\begin{credits}
\subsubsection{\ackname} We thank Jichen Li and Zihao Zhen for helpful discussions in the early stage of this work and several anonymous reviewers for valuable comments. This work is partially supported by Beijing Advanced Innovation Center for Future Blockchain and Privacy Computing.
\end{credits}


%
%
%
\bibliographystyle{splncs04}
\bibliography{mybibliography}

\appendix
\section{Missing Proofs for Section~\ref{s3}}\label{app}

\noindent
{\bf Theorem~\ref{thm1}.} (Restated)
{\em For any DAO $G$ and any grouping $\mathcal{G}=\{G^1, \dots, G^k\}$ of $G$, there exists a continuous grouping $\mathcal{\overline{G}}=\{\overline{G}^1, \dots, \overline{G}^k\}$
	such that $WTP_{total}(\mathcal{\overline{G}}) \geq WTP_{total}(\mathcal{G})$.}

\begin{proof}
	We show the existence of $\mathcal{\overline{G}}$ by construction. Recall that the members of $G$ are ordered by their valuations non-increasingly, $v_1\geq v_2\geq \cdots\geq v_n$, and that for each subgroup $G^j$ in $\mathcal{G}$, the members' valuations are also ordered non-increasingly,
	$v_1^j\geq \cdots \geq v_{n^j}^j$.
	We denote by  $CI^j = \arg\max_{i\in [n^j]}\{i\cdot v_i^j\}$ the critical index of $G^j$, breaking ties in favor of the largest $i$. The corresponding valuation is referred to as the critical value of $G^j$ and denoted by $CV^j$. Thus $WTP^j=CI^j\cdot CV^j$.
	
	We reorder the subgroups in $\mathcal{G}$ by their critical values from high to low, so that $CV^1 \ge CV^2 \ge \cdots \ge CV^k$.
	Since $1\leq CI^j\leq n^j$ for each $j\in [k]$,
	$$k\leq \sum_{j\in [k]} CI^j \leq \sum_{j\in [k]} n^j = n.$$
	Let the highest $CI^1$ members of the entire DAO $G$ form subgroup $\overline{G}^1$, the next $CI^2$ members of $G$ form subgroup $\overline{G}^2$, and so on until the subgroup $\overline{G}^{k-1}$. At last, put all the remaining members of $G$ into the final subgroup $\overline{G}^k$.
	That is, letting $N^j = \sum_{\ell=1}^j CI^\ell$ for each $j\in [k]$ and $N^0 = 0$, we have
	$$\overline{G}^j = \{N^{j-1}+1, N^{j-1}+2, \dots, N^j\} \ \forall j\in \{1, \dots, k-1\},$$
	and $\overline{G}^k = \{N^{k-1}+1, N^{k-1}+2, \dots, N^k, N^k+1, \dots, n\}$.
	
	It follows immediately that $\mathcal{\overline{G}}=\{\overline{G}^1, \dots, \overline{G}^k\}$ is a grouping of $G$ and $\mathcal{G}$ is continuous.
	The new grouping $\mathcal{\overline{G}}$ is still composed of k subgroups,
	$$
	\overline{n}^j=CI^j\ \forall j\in \{1, \dots, k-1\}\ \text{and}\ \overline{n}^k\ge CI^k.
	$$
	Next we prove that $\forall j \in [k]$, the valuation of subgroup $G^j$, $WTP^j$, is less than or equal to that of the newly formed subgroup $\overline{G}^j$, $\overline{WTP}^j$.
	(It should be noted that the $\overline{G}^j$'s here have not been reordered by their valuations, and only represent the subgroups defined in the above construction.)
	
	Case 1: $j < k$. By the definition of critical value, $\forall \ell\in [k]$, in the subgroup $G^\ell$ there are $CI^\ell$ members whose valuations are greater than or equal to $CV^\ell$.
	Since $CV^1 \ge CV^2 \ge \cdots \ge CV^j$ by our reordering of the subgroups in $\mathcal{G}$, in the entire DAO $G$, there are at least $\sum_{\ell=1}^{j}CI^\ell = N^j$ members whose valuations are greater than or equal to $CV^j$.
	That is,
	$$CV^j\leq v_{N^j}.$$
	Recall that $\overline{G}^j = \{N^{j-1}+1, N^{j-1}+2, \dots, N^j\}$ by construction.
	Let $\overline{v}_{\overline{n}^j}^j$ be the last member's valuation in $\overline{G}^j$, we thus have $\overline{v}_{\overline{n}^j}^j = v_{N^j}$ and
	$CV^j \leq \overline{v}_{\overline{n}^j}^j$.
	So for each $j<k$,
	$$\overline{WTP}^j \ge \overline{n}^j\cdot \overline{v}_{\overline{n}^j}^j = CI^j\cdot \overline{v}_{\overline{n}^j}^j \ge CI^j\cdot CV^j = WTP^j.$$
	
	Case 2: $j=k$. Similar to the previous case, there are at least $\sum_{\ell=1}^{k}CI^\ell = N^k$ members whose valuations are greater than or equal to $CV^k$, thus $CV^k\leq v_{N^k}$.
	Let $\overline{v}_{CI^k}^k$ be the $CI^k$-th highest valuation in $\overline{G}^k$, by the construction of $\overline{G}^k$
	we have $\overline{v}_{CI^k}^k = v_{N^k}$.
	Thus
	$CV^k \le
	\overline{v}_{CI^k}^k$
	and
	$$\overline{WTP}^k \ge CI^k\cdot \overline{v}_{CI^k}^k \ge CI^k\cdot CV^k = WTP^k.$$
	
	In sum, for each $j\in [k]$, $\overline{WTP}^j$ is greater than or equal to the corresponding $WTP^j$. Note that the $\overline{WTP}^j$'s and the $WTP^j$'s are not ordered non-increasingly.
	However, for each $j\in [k]$, the $j$-th highest value in $\{\overline{WTP}^1, \dots, \overline{WTP}^k\}$
	is greater than or equal to the $j$-th highest value in
	$\{WTP^1, \dots, WTP^k\}$.
	Thus, after reordering both the $\overline{WTP}^j$'s and the $WTP^j$'s non-increasingly,
	it still holds that
	$\overline{WTP}^j\geq WTP^j$ $\forall j\in [k]$.
	By the definition of $WTP_{total}$, we have $WTP_{total}(\overline{\mathcal{G}}) \ge WTP_{total}(\mathcal{G})$ as desired.
	$\hfill\square$
\end{proof}


\bigskip

\noindent
{\bf Theorem~\ref{thm2}.} (Restated)
{\em Given any DAO $G$, Algorithm~\ref{alg1} finds a continuous grouping $\mathcal{G}$ whose valuation is $opt\_WTP$ with time complexity $O(n^3)$.}

\begin{proof}
	By Corollary~\ref{cor1}, it suffices to maximize the DAO's willingness to pay among all continuous groupings.
	Therefore, the critical value $CV$ of the optimal grouping must appear as the valuation of a subgroup with members in a continuous interval. There are $O(n^2)$ possible continuous intervals for $[n]$, of the form $\{i, i+1, \dots, j\}$, with $i\in [n]$ and $i\leq j\leq n$.
	So there are $O(n^2)$ candidate $CV$s for the optimal grouping, enumerated in lines 2-6.
	
	Next, for each possible $CV$, one more loop of time $O(n)$ is needed to find its corresponding $CI$ and $WTP_{total}$, so the complexity is $O(n^3)$ to compute $opt\_WTP$ and the corresponding critical value $opt\_CV$ in lines 7-17.
	Another $O(n)$ time is needed to compute the optimal grouping given $opt\_CV$ in lines 18-28, thus the algorithm's total complexity is $O(n^3)$.
	
	It should be noted that the subgroups of $\mathcal{G}$ constructed in Algorithm~\ref{alg1} are not ordered by their $WTP^j$'s, but by their members' valuations non-increasingly.
	Since $WTP_{total}$ can be calculated directly through $CI\cdot CV$, in the proof below we will continue using this ordering of the subgroups.

	Assume, for the sake of contradiction, that Algorithm~\ref{alg1} does not find an optimal grouping, so there exists a continuous grouping $\overline{\mathcal{G}}$ with
	$$WTP_{total}(\overline{\mathcal{G}}) > WTP_{total}(\mathcal{G}).$$
	Denote $\overline{\mathcal{G}}$'s critical index by $\overline{CI}$ and critical value by
	$\overline{CV}$, then
	$$WTP_{total}(\overline{\mathcal{G}})=\overline{CI}\cdot \overline{CV}.$$
	Below we only consider the $\overline{CI}$ subgroups of $\overline{\mathcal{G}}$ whose valuations are at least $\overline{CV}$, as other subgroups do not contribute to $WTP_{total}(\overline{\mathcal{G}})$.
	Denote them by $\overline{G}^1, \dots, \overline{G}^{\overline{CI}}$ and their sizes by $\overline{n}^1,\dots,\overline{n}^{\overline{CI}}$.
	Notice that they are disjoint continuous intervals in $[n]$, but there may be gaps in between. Also that at least one of them has its valuation being exactly $\overline{CV}$.
	We reorder them
	by their members' valuations non-increasingly (or, by their members' indices non-decreasingly), so the lowest valuation in the $j^{th}$ subgroup $\overline{G}^j$
	is at most $v_{\sum_{i=1}^{j}{\overline{n}^i}}$, since  its index in $[n]$ is at least $\sum_{i=1}^{j}{\overline{n}^i}$.
	
	Because $\overline{CV}$ is the valuation of a continuous interval in $[n]$, the corresponding interval must have been enumerated in lines 2-6 of the algorithm and $\overline{CV}$ must have been added to the set $possible\_CV$.
	Wlog, we can assume that $\mathcal{G}$ and $k$
	output by Algorithm~\ref{alg1} are from when $CV$ is set to $\overline{CV}$ in lines 8-17, since otherwise the $WTP_{total}$ of the final output grouping can only get larger.
	Next, we'll find a contradiction following the assumption $WTP_{total}(\overline{\mathcal{G}}) > WTP_{total}(\mathcal{G})$.
	
	By construction, $WTP_{total}(\mathcal{G}) = k \cdot \overline{CV}$. Since
	$WTP_{total}(\overline{\mathcal{G}}) =\overline{CI}\cdot \overline{CV} > WTP_{total}(\mathcal{G})$,
	we have
	$$\overline{CI} \geq k+1.$$
	For each $j\in [k+1]$, let
	$\hat{n}^j$ be the size of $\hat{G}^j$ in the algorithm.
	According to the greedy way of constructing the subgroups in the algorithm, $G^1= \hat{G}^1$ is the smallest continuous interval starting from member $1$ with valuation at least $\overline{CV}$.
	Since $\overline{G}^1$ starts from a member index that is at least $1$, each of the member values in $\overline{G}^1$ is less than or equal to the one with the same ranking in $G^1$, and $\overline{G}^1$ cannot have less members to achieve valuation $\overline{CV}$. That is,
	$n^1= \hat{n}^1\leq \overline{n}^1$.
	
	Inductively, assume the sizes of the first $j<k$ subgroups of $\mathcal{G}$ are all less than or equal to those of the corresponding ones of $\overline{\mathcal{G}}$;
	that is,
	$n^i = \hat{n}^i \leq \overline{n}^i \ \forall 1\leq i \leq j$.
	Again by the greedy construction of the algorithm, $\hat{G}^{j+1}$ is the smallest continuous interval starting from member $\sum_{i=1}^{j} \hat{n}^i +1$ with valuation at least $\overline{CV}$. By the inductive hypothesis, $\overline{G}^{j+1}$ starts from a member index that is at least $\sum_{i=1}^j \overline{n}^i +1 \geq \sum_{i=1}^{j} \hat{n}^i +1$.
	So each of the member values in $\overline{G}^{j+1}$ is less than or equal to the one with the same ranking in $\hat{G}^{j+1}$, and
	$\hat{n}^{j+1}\leq \overline{n}^{j+1}$ as before.
	Thus by induction we have
	$\hat{n}^j \leq \overline{n}^j \ \forall j\in [k]$
	and
	$$\sum_{j=1}^k \overline{n}^j \geq \sum_{j=1}^k \hat{n}^j.$$
	
	By construction, $\hat{G}^{k+1}$ contains all members from $\sum_{j=1}^k \hat{n}^j +1$ to $n$. On the other hand, $\overline{G}^{k+1}$ starts from a member index that is at least $\sum_{j=1}^k \overline{n}^j +1 \geq \sum_{j=1}^k \hat{n}^j +1$, and ends at a member index that is at most $n$.
	Accordingly,
	$$\overline{G}^{k+1}\subseteq \hat{G}^{k+1}.$$
	
	
	Again by the construction of the algorithm, $\hat{G}^{k+1}$'s valuation is strictly less than $\overline{CV}$, as otherwise the corresponding critical index would have been $k+1$ instead of $k$.
	Thus so does $\overline{G}^{k+1}$ ---that is, they only consist of the ``useless part'' that cannot form a subgroup to reach the desired critical value.
	This is a contradiction to the condition $\overline{CI} \geq k + 1$ which follows from the assumption $WTP_{total}(\overline{\mathcal{G}}) > WTP_{total}(\mathcal{G})$. Therefore the assumption is false and Theorem \ref{thm2} holds.
	$\hfill\square$
\end{proof}

\bigskip

\noindent
{\bf Theorem~\ref{thm3}.} (Restated)
{\em For a winning DAO $G$,
	the optimal grouping $\mathcal{G}$ by Algorithm~\ref{alg1}
	is less-excludable than the degenerated grouping $\mathcal{G}^*$.
}

\bigskip
To prove Theorem \ref{thm3} we first have the following lemma.

\begin{lemma}\label{lemma1}
	For a winning DAO $G$ and the optimal grouping $\mathcal{G}=\{\hat{G}^1, \dots, \hat{G}^k\cup \hat{G}^{k+1}\}$ by Algorithm~\ref{alg1}, we have $\hat{n}^1 \le \hat{n}^2 \le \cdots \le \hat{n}^k$, where $\hat{n}^j=|\hat{G}^j|$ $\forall j\in [k]$.
\end{lemma}

\begin{proof}
	Assume that for the sake of contradiction, $\exists a,b \in [k],\ a<b \ \mbox{and} \ \hat{n}^a > \hat{n}^b$.
	Since $\hat{n}^b\cdot v^b_{\hat{n}^b} \geq opt\_CV$ by construction, and since the members are ordered by their valuations non-increasingly, we have $v^a_{\hat{n}^b}\geq v^b_{\hat{n}^b}$ and
	then $$\hat{n}^b\cdot v_{\hat{n}^b}^a \ge \hat{n}^b\cdot v_{\hat{n}^b}^b \ge opt\_CV.$$
	Therefore, when constructing $\hat{G}^a$, by its greedy nature Algorithm~\ref{alg1} stops as soon as the current subgroup's valuation is at least $opt\_CV$, which means it should stop at having $\hat{n}^b$ members instead of $\hat{n}^a$ members, a contradiction. 
	$\hfill\square$
\end{proof}

\begin{proof}[of Theorem \ref{thm3}]
	As long as $G$ is the winning DAO,
	the price $P_{total}$ is the second highest bid in the auction and is unrelated to the substructure of $G$.
	Thus $P_{total}$ doesn't change whether the underlying grouping is $\mathcal{G}^*$ or $\mathcal{G}$. Also, $$P_{total}\leq WTP_{total}(\mathcal{G}^*) \leq WTP_{total}(\mathcal{G}) = opt\_WTP.$$
	
	
	
	According to $\mathcal{G}^*$, the first $i^*$ bidders in $G$ get access to the item, where $i^*=\max_{i\in[n]}\{i | i \cdot v_i\ge P_{total}\}$. In particular,
	$$i^* \cdot v_{i^*}\geq P_{total}.$$
	Next we prove that these $i^*$ bidders also get access according to $\mathcal{G}$.
	According to Algorithm~\ref{alg1},
	$$opt\_WTP=k\cdot opt\_CV \ge P_{total} \text{ and } WTP^j \ge opt\_CV \ \forall j \in [k].$$
	Thus $k\cdot \min_{j\in [k]} WTP^j \geq k\cdot opt\_CV \geq P_{total}$,
	all $k$ subgroups have access to the item, and each subgroup pays $p= \frac{P_{total}}{k}$.
	We distinguish two cases.
	
	Case 1: member $i^*$ of $G$ is not in the final subgroup $G^k$ in Algorithm~\ref{alg1}. Since $\forall j \in \{1, \dots, k-1\}$, $G^j=\hat{G}^j$, by the construction of Algorithm~\ref{alg1} we have
	$WTP^j=\hat{n}^j\cdot v_{\hat{n}^j}^j \ge opt\_CV \geq \frac{P_{total}}{k}$.
	Thus all members in these subgroups get access, which implies that the first $i^*$ members in $G$ get access.
	
	Case 2: member $i^*$ of $G$ is in the final subgroup $G^k=\hat{G}^{k}\cup \hat{G}^{k+1}$ in Algorithm~\ref{alg1}.
	We denote its rank in the subgroup as $i_1$, so $v_{i^*}=v_{i_1}^k$.
	Since $$\hat{n}^k\cdot v_{\hat{n}^k}^k \ge opt\_CV \ge \frac{P_{total}}{k},$$ the first $\hat{n}^k$ bidders in the subgroup $G^k$ also have access to the item. So in the following proof, we only need to consider the case where $i_1 > \hat{n}^k$.
	
	According to Lemma~\ref{lemma1}, we have
	$$\hat{n}^1 \le \cdots \leq \hat{n}^{k-1} \le \hat{n}^k \leq  i_1.$$
	Moreover, $\forall j \in [k-1],\ n^j = \hat{n}^j,$
	where $n^j = |G^j|$.
	So, $i_1 \ge n^{k-1} \ge \frac{\sum_{j=1}^{k-1}n^j}{k-1}$,
	which means
	$$i^*=\sum_{j=1}^{k-1}n^j+i_1 \ge \sum_{j=1}^{k-1}n^j+\frac{\sum_{j=1}^{k-1}n^j}{k-1}=\frac{\sum_{j=1}^{k-1}n^j}{k-1}\cdot k.$$
	So $$\frac{P_{total}}{k}\cdot(k-1)\ge \frac{P_{total}}{i^*}\cdot \sum_{j=1}^{k-1}n^j.$$
	It represents that for the bidders in the first $k-1$ subgroups, their total payment in the optimal grouping $\mathcal{G}$
	is at least their total payment in the original degenerated grouping $\mathcal{G}^*$.
	
	So in the final subgroup $G^k$, its payment is
	$$\frac{P_{total}}{k}=P_{total}-\frac{P_{total}}{k}\cdot(k-1)\le P_{total}-\frac{P_{total}}{i^*}\cdot\sum_{j=1}^{k-1}n^j = \frac{P_{total}}{i^*}\cdot i_1 \le v_{i^*}\cdot i_1=i_1\cdot v_{i_1}^k.$$
	Thus in $G^k$, at least the first $i_1$ bidders get access to the item.
	
	Combining the two cases together, all $i^*$ bidders that get access according to $\mathcal{G}^*$ still
	get access according to $\mathcal{G}$, thus $\mathcal{G}$ is less-excludable.
	$\hfill\square$
\end{proof}

\section{Missing Proofs for Section~\ref{s4}}\label{app2}

In Algorithm~\ref{alg3}, the output is a feasible group of parameters with the widest pool. It goes through different pool widths that start with $k_2=n$ and continue to decrease. Based on the condition that $P_1=(1+\alpha)P_2\leq (1+\alpha)b_{k_2}$, it calculates the maximum capacity of each pool until it reaches $P_{total}$. As lines 1-5 show, this part is a double loop of time $O(n^2)$ and stops no later than when $size$ reaches $WTP_{total}$.
After determining the pool width $k_2$, the height of the ``water line'' $P_1$ is determined by fulfilling the pool with the amount of water $P_{total}$ layer by layer. As lines 6-16 show, it consists of a loop of time $O(n)$. In conclusion, the time complexity of Algorithm~\ref{alg3} is $O(n^2)$.

\medskip
\noindent
{\bf Theorem~\ref{thm4}.} (Restated)
{\em Under the collective-utility model and $\forall \alpha> 0$, mechanism $M_1$ satisfies individual rationality, budget balance, and equal treatment.
	It is also incentive compatible in generic games.
	Given a winning DAO $G$,
	$M_1$ achieves a higher willingness to pay for $G$, a higher social welfare, and a less-excludable allocation compared with $M_\ell$.}

\bigskip
To prove Theorem \ref{thm4} we first have the following lemma, which shows that the feasible parameters with the largest $k_2$ is the best to look for.

\begin{lemma}\label{lemma2}
	Given $P_{total}$ and $b_1, \dots, b_n$, and given two groups of feasible parameters $(P_1, P_2, k_1, k_2)$ and $(P_1', P_2', k_1', k_2')$, we denote the allocation vectors by $x$ and $x'$, and the utility vectors
	by $u$ and $u'$, respectively.
	If $k_2'<k_2$, then $P_1'>P_1$, $P_2'>P_2$, $k_{1}'<k_{1}$, $x$ is less-excludable than $x'$ and $\forall i\in[n]$, $u_i\geq u_i'$.
\end{lemma}


\begin{proof}
	Given $k_2'<k_2$, for the sake of contradiction, assume that $P_1'\leq P_1$, then according to the definition of $k_1$ and $k_2$, we have $k_{2}>k_{2}'\geq k_{1}'\geq k_{1}$. Since the parameters are feasible, budget balance should be satisfied, we have $P_{total}=k_{1}\cdot P_1+\sum_{i=k_{1}+1}^{k_{2}}b_i>k_{1}\cdot P_1+\sum_{i=k_{1}+1}^{k_{1}'}b_i+\sum_{i=k_{1}'+1}^{k_{2}'}b_i\geq k_{1}'\cdot P_1'+\sum_{i=k_{1}'+1}^{k_{2}'}b_i=P_{total}$, which is a contradiction. Thus $P_1'> P_1$, $k_{1}'< k_{1}$, $P_2'=\frac{P_1'}{1+\alpha}>\frac{P_1}{1+\alpha}=P_2$.
	
	According to the access function, for any member $i\leq k_2'$, $x_i'=\min(1, \frac{b_i}{P_1'})\leq \min(1,\frac{b_i}{P_1})=x_i$, $u_i'=v_i\cdot x_i'-\min(b_i,P_1')+\alpha\sum_{j=1}^{k_2'}(v_j\cdot x_j')\leq v_i\cdot x_i-min(b_i,P_1)+\alpha\sum_{j=1}^{k_2'}(v_j\cdot x_j)=u_i$; for any member $k_2'<i\leq n$, $x_i'=0\leq x_i$, $u_i'=0\leq u_i$. Thus, $x_i\geq x_i'$ and $u_i\geq u_i'$ for any member $i\in[n]$.
	$\hfill\square$
\end{proof}

\medskip

\begin{proof}[of Theorem \ref{thm4}]
	We divide the proof into six parts.
	\\
	
	\noindent{\em 1. Individual rationality.}
	
	For any bidder $i$ in DAO $G$, consider the scenario when $b_i=v_i$.
	We assume $G$ is the winning DAO, otherwise, $u_i=0$. Consider the following cases.
	
	\begin{itemize}
		\setlength{\itemindent}{2em}
		\item[Case 1:] $i\leq k_1$. In this case, $v_i\geq P_1$, $x_i=1$, $p_i=P_1$. Thus, $u_i=(1+\alpha)v_ix_i-p_i+ \alpha \sum_{j\neq i}{(v_j\cdot x_j)}\geq (1+\alpha)v_i-P_1 \geq 0$.
		
		\item[Case 2:] $k_1<i\leq k_2$. In this case, $P_1>v_i\geq \frac{P_1}{1+\alpha}$, $x_i=\frac{v_i}{P_1}$, $p_i=v_i$. Similarly,  we have $u_i\geq(1+\alpha)v_ix_i-p_i=(1+\alpha)v_i\cdot\frac{v_i}{P_1}-v_i =\frac{v_i}{P_1}[(1+\alpha)v_i-P_1]\geq 0$.
		
		\item[Case 3:] $i>k_2$. In this case, $x_i=0$, $p_i=0$, thus $u_i=0$.
	\end{itemize}
	In conclusion, $\forall i \in [n]$, $\forall b_{-i}$, $u_i(v_i, b_{-i})\geq 0$, individual rationality holds.
	\\
	
	\noindent{\em 2. Budget balance.}
	
	We only need to consider the winning DAO. According to Algorithm~\ref{alg3} and the cost-sharing function, $P_{total}=k_{1}\cdot P_1+\sum_{i=k_{1}+1}^{k_{2}}b_i =\sum_{i=1}^{k_{2}}p_i$, satisfying budget balance.
	\\
	
	\noindent{\em 3. Equal treatment.}
	
	For any two members $i\neq j$ with $b_i=b_j$, wlog, we assume $i<j$.
	
	\begin{itemize}
		\setlength{\itemindent}{2em}
		\item[Case 1:] $j\leq k_1$. In this case, $b_i=b_j\geq P_1$, so $i\leq k_1$, $x_i=x_j=1$, $p_i=p_j=P_1$.
		
		\item[Case 2:] $k_1<j\leq k_2$. In this case, $P_2\leq b_i=b_j< P_1$, so $k_1<i\leq k_2$, $x_i=\frac{b_i}{P_1}=\frac{b_j}{P_1}=x_j$, $p_i=b_i=b_j=p_j$.
		
		\item[Case 3:] $j> k_2$. In this case, $b_i=b_j<P_2$, so $i>k_2$, $x_i=x_j=0$, $p_i=p_j=0$.
	\end{itemize}
	Thus, if $b_i=b_j$, then $x_i=x_j$ and $p_i=p_j$, $\forall i,j\in[n]$, $i\neq j$. Namely, $M_1$ satisfies equal treatment.
	\\
	
	\noindent{\em 4. Underbid is weakly dominated by the truthful bid.} 
	
	For any bidder $i$ in DAO $G$, for any fixed $b_{-i}$, we analyze two scenarios: Scenario~1 refers to $i$ bidding truthfully, $b_i=v_i$; Scenario~2 is that $i$ intentionally lowers its bid to $b_i' = v_i - \delta$ with $0 < \delta < v_i$. Let the access, payment and utility of $i$ in these two scenarios be denoted as $x_i$, $p_i$, $u_i$ and $x_i'$, $p_i'$, $u_i'$. The willingness to pay of the DAO is denoted as $WTP_{total}$ and $WTP_{total}'$, respectively.
	
	According to the aggregation function, $WTP_{total}\geq WTP_{total}'$. Thus, we only need to consider the case where $G$ can win in both scenarios, otherwise, by individual rationality proved above, we have $u_i\geq 0=u_i'$, underbid is dominated.
	
	Since $b_{-i}$ is fixed, the required bid for DAO $G$ to win and its total payment $P_{total}$ after winning remain unchanged.
	And the utility of $i$ can be seen as a function of its own bid and $b_{G\setminus\{i\}}$, where $b_{G\setminus\{i\}}$ is the bids of all members of $G$ except for $i$.
	We further denote the allocation parameters output by Algorithm~\ref{alg3} after winning as $(P_1, P_2, k_1, k_2)$ and $(P_1', P_2', k_1', k_2')$, respectively, which also represent the widest feasible pool in each scenario. We distinguish three cases.
	\begin{itemize}
		\setlength{\itemindent}{2em}
		\item[Case 1:] $i\leq k_1$.
		In Scenario~1, $i$ can gain full access, $v_i\geq P_1=p_i$, $x_i=1$, $u_i=v_i\cdot1-P_1+\alpha \sum_{j=1}^{k_1}v_j+\alpha\sum_{j=k_1+1}^{k_2}(v_j\cdot\frac{b_j}{P_1})$.
		According to individual rationality, we have $u_i \geq 0$, so in Scenario~2, there is no need to consider the case where $b_i'$ is too small to obtain any access and thus $u_i'=0$.
		If $k_2'<k_2$, for $P_{total},\ b_i,\ b_{G\setminus\{i\}}$ in Scenario~1, $b_i\geq b_i'$ and others' bids remain the same as in Scenario~2, Algorithm~\ref{alg3} should output $k_2'$ instead of $k_2$ as the width of the widest feasible pool, which is a contradiction.
		Thus, $k_2'\geq k_2$ and there remain two subcases for Scenario~2.
		
		\begin{itemize}
			\setlength{\itemindent}{3em}
			\item[Subcase 1.1:] $b_i' \geq P_1$.
			In this case, $P_{total}=k_1\cdot P_1+\sum_{j=k_1+1}^{k_2}b_j$, so the former parameter in Scenario~1 is still feasible in Scenario~2. It is also the optimal parameter, otherwise, a wider pool with width $k_2'>k_2$ is feasible in Scenario~1. So we have $(P_1', P_2', k_1', k_2')=(P_1, P_2, k_1, k_2)$, $b_i'\geq P_1'$, $x_i' = 1$, $u_i' = u_i$, underbid is meaningless.
			
			\item[Subcase 1.2:] $b_i' < P_1$.
			In this case, given $P_{total},\ b_i',\ b_{G\setminus\{i\}}$, from Lemma~\ref{lemma2}, we know that as the selected $k_2'$ decreases, $P_1'$ must increase. Therefore, $P_1'$ reaches its minimum when $k_2' = k_2$.
			Even if $k_2' = k_2$, we have $P_{total} = k_1 \cdot P_1 + \sum_{j=k_1+1}^{k_2} b_j > (k_1-1) \cdot P_1 + b_i' + \sum_{j=k_1+1}^{k_2} b_j$. To be feasible, $P_1'$ must be greater than $P_1$.
			So, $b_i' < P_1'$, $x_i' = \frac{b_i'}{P_1'} < 1$.
			By access function, for any other member $j$ in the DAO, we have $x_j' \leq x_j$ since its bid remains unchanged but the ``water line'' rises, the amount of access it receives will not increase.
			Thus, $u_i'=v_i\cdot x_i'-P_1'+\sum_{j=1}^{n}v_j\cdot x_j'<v_i-P_1+\sum_{j=1}^{n}v_j\cdot x_j=u_i$.
		\end{itemize}
		
		\item[Case 2:] $k_1<i\leq k_2$.
		In Scenario~1, $i$ can gain partial access, $v_i < P_1$, $x_i<1$, $p_i=v_i$, $$u_i=(1+\alpha)v_i\cdot\frac{v_i}{P_1}-v_i+\alpha \sum_{j=1}^{k_1}v_j+\alpha\sum_{\substack{j=k_1+1\\j\neq i}}^{k_2}(v_j\cdot\frac{b_j}{P_1}).$$
		In Scenario~2, similar to the previous case where $b_i' < P_1$, we can prove that $k_2'\leq k_2$, $P_1'>P_1$ and for all member $j$ in $G$, $x_j'$ will not increase as long as $P_1'$ increases. The utility $u_i'=v_i\cdot\frac{b_i'}{P_1'}-b_i'+\alpha \sum_{j=1}^{n}v_j\cdot x_j'$.
		
		Therefore, when fixing $b_i'$, and considering $u_i'$ as a function of $P_1'$, $u_i'$ is monotonically decreasing with respect to $P_1'$.
		To prove that $u_i'$ will not exceed $u_i$, we only need to consider the case when $P_1'$ is at its minimum.
		Given $P_{total}$, $b_i'$, and $b_{G\setminus\{i\}}$, from Lemma~\ref{lemma2}, we know that as $k_2'$ decreases, $P_1'$ must increase. Therefore, assume that $k_2'$ takes the maximum value $k_2$. From the definition of $k_1$, a decrease in $k_1'$ means an increase in $P_1'$. Since $P_1' > P_1>v_i$, we have $k_1' \leq k_1$. Hence, we assume that $k_1'$ takes the maximum value $k_1$.
		
		In conclusion, we only consider where $k_1'= k_1$, $k_2'= k_2$, thus $$u_i'=(1+\alpha)v_i\cdot\frac{b_i'}{P_1'}-b_i'+\alpha \sum_{j=1}^{k_1}v_j+\alpha\sum_{\substack{j=k_1+1\\j\neq i}}^{k_2}(v_j\cdot\frac{b_j}{P_1'}).$$
		Through $P_{total}=k_1\cdot P_1+v_i+\sum_{\substack{j=k_1+1\\j\neq i}}^{k_2}b_j=k_1\cdot P_1'+b_i'+\sum_{\substack{j=k_1+1\\j\neq i}}^{k_2}b_j$, we have $P_1'= P_1+\frac{b_i'-v_i}{k_1}=P_1+\frac{\delta}{k_1}$. Thus, we denote $ \Delta u_i(\delta)$ and $f(\delta)$ as $ \Delta u_i(\delta)=u_i'-u_i\leq (1+\alpha)[\frac{v_i(v_i-\delta)}{P_1+\frac{\delta}{k_1}}-\frac{v_i^2}{P_1}]+\delta=\frac{f(\delta)}{P_1 (P_1+\frac{\delta}{k_1})}$. Then we have $f(\delta)=\frac{P_1}{k_1}\delta^2+[P_1^2-(1+\alpha)v_iP_1-\frac{(1+\alpha)v_i^2}{k_1}]\delta$, which is a quadratic function. So, to prove $f(\delta)\leq 0$, $\forall \delta \in (0,v_i)$, it is sufficient to prove $f(v_i)\leq 0$.
		While $f(v_i)=v_i[\frac{P_1}{k_1}v_i+P_1^2-P_1(1+\alpha)v_i-\frac{1+\alpha}{k_1}v_i^2]=v_i[P_1-(1+\alpha)v_i](P_1+\frac{v_i}{k_1})$, due to $i\leq k_2$, we have $v_i\geq P_2=\frac{P_1}{1+\alpha}$, which means $P_1-(1+\alpha)v_i\leq 0$, so $f(v_i)\leq 0$.
		
		\item[Case 3:] $i>k_2$.
		In Scenario~1, $i$ fails to gain access, $v_i<P_2$, $x_i=0$, $u_i=0$. Thus, $\forall b_i'<v_i$, $x_i'=0$, $u_i'=0$, underbid is meaningless.
	\end{itemize}
	Combining the above three cases, we have $u_i(v_i, b_{-i}) \geq u_i(b_i', b_{-i})$, $\forall b_i'< v_i$, $\forall i\in [n]$, $\forall b_{-i}$.
	\\
	
	\noindent{\em 5. Overbid is weakly dominated by the truthful bid in generic games.}
	
	We follow the notation above except for $b_i'=v_i+\delta$ with $\delta>0$ in Scenario~2. According to the aggregation function, $WTP_{total}'\geq WTP_{total}$, so we assume $G$ wins in Scenario~2, otherwise, $u_i'=u_i=0$.
	Below we consider two cases.
	
	\begin{itemize}
		\setlength{\itemindent}{2em}
		\item[Case 1:] $G$ wins in Scenario~1.
		After overbid, the pool in Scenario~1, whose width is $k_2$, is still feasible to accommodate the total amount of $P_{total}$ in Scenario~2.
		If $i\leq k_1$, we can prove that $k_2'=k_2$, the group of optimal parameters remains the same, so the allocation and utility remain the same.
		But if $i > k_1$, the overbid may help a wider pool increase its capacity and become feasible. On this occasion, Algorithm~\ref{alg3} will choose the widest pool, so the output will be $k_2' \geq k_2$ and $u_i'$ may not be the same as $u_i$.
		Thus, we consider the case where $i > k_1$ and divide it into three subcases.
		
		\begin{itemize}
			\setlength{\itemindent}{3em}
			\item[Subcase 1.1:] $i \leq k_1'$.
			In this case, $x_i'=1$, $p_i'=P_1'$. If $p_i' > v_i$, then $u_i'=-\infty$, otherwise, $p_i' \leq v_i<b_i'$. Since $P_{total}$ remains unchanged, $(P_1', P_2', k_1', k_2')$ is still feasible in Scenario~1 where $p_i=p_i'\leq v_i$, $i$ could gain full access, contradicting the condition $i > k_1$ which means $i$ cannot gain full access in Scenario~1.
			
			\item[Subcase 1.2:] $k_1' < i \leq k_2'$. In this case, $i$ has to pay $p_i' = b_i' > v_i$, thus $u_i'=-\infty$.
			
			\item[Subcase 1.3:] $i > k_2'$. In this case, $x_i'=0$ and $u_i'=0$, overbid is meaningless.
		\end{itemize}
		
		\item[Case 2:] $G$ loses in Scenario~1.
		We denote the winning bid in Scenario~1 as $WTP_{1}$. In generic games, there is no critical tie and we have $WTP_{total}'> WTP_{1}> WTP_{total}$.
		In Scenario~1, $u_i=0$.
		In Scenario~2, $G$ should pay $WTP_{1}$ as it becomes the second highest bid.
		If $p_i'>v_i$, then $u_i'=-\infty$ and we are done.
		Otherwise, $p_i'\leq v_i<b_i'$ and we have $k_1'<i\leq k_2'$. According to the aggregation function, we can prove that in Scenario~1, $WTP_{total}\geq\sum_{j=1}^{k_2'}\min\{b_j, (1+\alpha)b_{k_2'}\}\geq\sum_{j=1}^{k_2'}\min\{b_j, P_1'\}=WTP_{1}$, which is also a contradiction.
	\end{itemize}
	In sum, in generic games, overbid is weakly dominated by the truthful bid.
	\\
	
	\noindent{\em 6. $WTP_{total}$, $SW$ and excludability in comparison with $M_\ell$.}
	
	Denote that in the mechanism $M_\ell$, DAO $G$'s willingness to pay is $WTP_{total}'=i_0\cdot v_{i_0}$, $i_0\in[n]$. In the new mechanism $M_1$, $WTP_{total}\geq \sum_{i=1}^{i_0}\min\{v_i,(1+\alpha)v_{i_0}\}\geq i_0\cdot v_{i_0}=WTP_{total}'$.
	
	Denote $i^*_0 = \max_{i\in[n]} \{i | i \cdot v_i \geq P_{total}\}$ in the mechanism $M_\ell$. Then, in Algorithm~\ref{alg3}, using $i^*_0$ as the width limit, we can find a pool whose capacity is greater than or equal to $P_{total}$. At this point, $P_1 \leq v_{i^*_0}$, and $k_2 = i^*_0$. This is in fact the allocation from the mechanism $M_\ell$, denoted by $x'$. If the allocation $x$ obtained in the mechanism $M_1$ by Algorithm~\ref{alg3} is different from $x'$, it must correspond to a wider pool. According to Lemma~\ref{lemma2} and the access function, it is less-excludable than the allocation $x'$. Therefore, $\forall i \in [n], x_i \geq x_i'$,
	and $SW = \sum_{i=1}^{n} v_i \cdot x_i \geq \sum_{i=1}^{n} v_i \cdot x_i' = SW'$.
	
	\medskip
	Combining the above six parts together, Theorem \ref{thm4} holds.
	$\hfill\square$
\end{proof}




\end{document}